\documentclass[11pt, a4paper]{article}
\usepackage[x11names]{xcolor}
\usepackage[a4paper, margin=1in]{geometry}
\usepackage[english]{babel}

\usepackage{authblk}
\usepackage{amsmath, amssymb, amsthm}
\usepackage{booktabs}
\usepackage{enumitem}
\usepackage{hyperref}

\hypersetup{
    colorlinks=true,
    linkcolor=blue!70!black,
    citecolor=blue!70!black,
    urlcolor=blue!70!black
}

\theoremstyle{plain}
\newtheorem{theorem}{Theorem}[section]

\newtheorem{claim}[theorem]{Claim}

\newtheorem{proposition}[theorem]{Proposition}

\theoremstyle{definition}

\theoremstyle{remark}

\title{\textbf{Building On Learning-Augmented Online Sorting}}

\author{Wyatt Olson}
\affil{Gianforte School of Computing, Montana State University\\
\texttt{wyatt.olson@montana.edu}}

\date{} 

\begin{document}

\maketitle

\begin{abstract}
\noindent Online sorting is a problem that has gained attention recently, and is a natural translation of sorting into an online setting. The problem at hand is to place elements of an online sequence immediately and irrevocably with the goal of producing an array that is as close to sorted as possible upon inserting all elements. Recent work by Bercea \textit{et al.} \cite{bercea2026online} introduced the incorporation of learned predictions into the online sorting domain. They show lower bounds for the consistency and robustness of any algorithm, including randomization, then provide a deterministic algorithm that is $O(\frac{n}{r})$-consistent and $O(r)$-robust, which are tight bounds. However, this algorithm requires prior knowledge of the optimal offline cost of the full input sequence before all elements have been revealed. We resolve this by providing an algorithm with matching bounds that does not require knowledge of the optimal cost, then extend these results to the more general problem of online metric TSP, in line with the prior work.
\end{abstract}

\section{Introduction}
\label{sec:intro}
A small but interesting line of work has recently arisen surrounding the problem of online sorting, which was first proposed by Aamand \textit{et al.} \cite{aamand2023online} in work exploring translation-only polygon packing in an online setting. The problem is as follows: we are given an online input sequence $X=x_1,x_2, \dots, x_n$ to be sorted into an array $A$ of length $n$ such that an element of the input sequence must be irrevocably placed into an empty cell of $A$ before the next element of the sequence appears. The goal is for $A$ to be as close to sorted as possible upon assigning all values of $X$ to cells in $A$. The cost of the resulting array $A$ provides a measure to quantify this closeness and is defined as $c = \sum_{i=1}^{n-1}|A[i]-A[i+1]|$. It is clear that the optimal cost is given when $A$ is sorted in the traditional sense, in which case the cost is simply the range of the elements of $X$ (\textit{i.e.} $c_{opt}=\mathrm{range}(X)$). 

\subsection{Related Work}
\label{subsec:related}

Subsequent work has defined optimal bounds for the cost in a setting entirely oblivious to the input sequence $X$. We say that an algorithm is $\rho$-competitive if the cost $c$ is at most $\rho\cdot c_{opt}$ over all instances of the problem. Abrahamsen \textit{et al.} \cite{abrahamsen2024online} provide an algorithm with competitive ratio of $O(\sqrt{n})$ without knowledge of the range of $X$, and show a matching lower bound for randomized algorithms, building on the same bound provided for deterministic algorithms by Aamand \textit{et al.} Of particular interest to us is the work by Bercea \textit{et al.} \cite{bercea2026online} introducing the integration of learned predictions into the online sorting problem. In this scenario, we receive a predicted index denoted $\hat{i}_t$, for each value $x_t$, where $\hat{I}=\hat{i}_1, \hat{i}_2, \dots, \hat{i}_n$ is a permutation of the indices of $A$. That is to say, each prediction in the sequence is unique, and assigning each $x_t$ to the corresponding predicted cell $A[\hat{i}_t]$ results in no two elements in $X$ being placed in the same cell of $A$. 

With the introduction of predictions to an algorithm, we utilize two distinct metrics to analyze the cost in place of competitiveness. We say that an algorithm is $\alpha$-consistent if for all instances of the problem we are guaranteed a cost at most $\alpha\cdot\hat{c}$, where $\hat{c}$ is the cost of an algorithm that simply follows the provided predictions. Following convention, we also say that an algorithm is $r$-robust if for all instances of the problem we are guaranteed $r$-competitiveness, regardless of predictions and predicted cost. Intuitively, robustness ensures bad predictions don't result in excessively poor cost, while consistency ensures that to some degree, we are able to reap the rewards of good predictions. For online sorting, Bercea \textit{et al.} \cite{bercea2026online} provide tight bounds for both metrics: $O(r)$-robustness and $O(\frac{n}{r})$-consistency for $1\leq r\leq n$ (the lower bound tightens this to $\sqrt{n}\leq r\leq n$), which includes randomized algorithms. Of interest to us, they provide a deterministic algorithm that meets these bounds for online sorting; however, it requires prior knowledge of the optimal offline cost to meet these bounds. Our contribution is an algorithm inspired by this work that meets the same bounds without knowledge of the optimal cost. 

Another contribution of Bercea \textit{et al.} is providing an asymptotically tight learning-augmented algorithm for a more generalized version of online sorting, which was presented in \cite{abrahamsen2024online} as the online TSP problem. The bounds for this general setting match those of sorting, and the algorithmic approaches are nearly identical from the perspective of algorithms utilizing predictions. One appropriate distinction to make in this problem is replacing our prior $O(\sqrt{n})$-competitive solution with the clever algorithm devised by Bertram in \cite{bertram2025online} that handles the complications involved with distances in higher-dimensional metric spaces to match the $O(\sqrt{n})$-competitive algorithm for online sorting. Again, the algorithm presented in \cite{bercea2026online} has the downside of requiring knowledge of the optimal cost of the input sequence. Similar to the prior work, we are able to treat the $O(\sqrt{n})$-competitive algorithms for the respective problems effectively as black boxes, so our approach for online sorting can be transferred neatly to this general space, matching the prior work's bounds without prior knowledge of the optimal cost. 

\subsection{Contributions}
\label{subsec:contributions}
In this work we provide a pair of algorithms that match the cost bounds of those given by Bercea \textit{et al.} \cite{bercea2026online} for both online sorting and online metric TSP without requiring knowledge of the optimal solution to ensure those bounds. We choose to separate these two cases despite the latter being a generalization of the former because online sorting provides a much more natural space to work in, allowing us to build intuition that eases the understanding of online metric TSP. 

\section{Online Sorting With Predictions}
\label{sec:algorithm}

Here we present an alternative algorithm to that given by \cite{bercea2026online} that does not require knowledge of the optimal cost to achieve both $O(\frac{n}{r})$-consistency and $O(r)$-robustness. Similar to this prior work, it utilizes a block structure and the deterministic $O(\sqrt{n})$-competitive algorithm presented in \cite{abrahamsen2024online} that is oblivious to predictions as a sub-routine; however, we run an independent instance of this algorithm within each block with a coarser block width, preventing the need for a second phase or monitoring the cost with respect to the optimal. 

\subsection{Defining the Algorithm}
Let $r$ be a parameter $\sqrt{n}\leq r\leq n$, following from prior work \cite{bercea2026online}. We break up the target array $A$ into $k = \frac{r^2}{n}$ contiguous blocks of size $m=(\frac{n}{r})^2$, where the $j$-\textit{th} block consists of the $((j-1)\cdot m+1)$-st cell to the $(j\cdot m)$-th cell. Following \cite{bercea2026online}, to simplify notation we assume without loss of generality that $m$ and $k$ are integers. Let $X=x_1, x_2, \dots, x_n$ be our input sequence to be sorted and $\hat{I}=\hat{i}_1, \hat{i}_2, \dots, \hat{i}_n$ be the sequence of predictions such that at time $t$, item $x_t$ would be inserted at $A[\hat{i}_t]$ if we were to simply follow the sequence of predictions. 

Our algorithm is simple when abstracting the online sorting algorithm from \cite{abrahamsen2024online}; let us refer to it as \textsc{Alg}$_0$. For each block, create a new instance of \textsc{Alg}$_0$ that will run on the sequence of items we assign to the given block. Then, for each input $x_t$, identify the block $B_j$ corresponding to its predicted index; that is, find $B_j$ such that $A[\hat{i}_t]\in B_j$. Hand off $x_t$ to the instance of \textsc{Alg}$_0$ corresponding to $B_j$ that will in turn place $x_t$ within $B_j$. We note here that we are guaranteed each block receives precisely the number of elements it needs because the predictions are a permutation of the indices of $A$. 

\subsection{Cost Analysis}
We now analyze the cost of the solution of this algorithm, aiming to show that it is $O(\frac{n}{r})$-consistent and $O(r)$-robust. We first reiterate a fact shown in \cite{bercea2026online}, referenced there as Claim 6. 

\begin{proposition}[Adapted from Bercea et al.~\cite{bercea2026online}, Claim 6]
\label{proposition:block-range-bound}
Let $A$ be an array partitioned into $k$ blocks $B_1, \dots, B_k$. 
Then the predicted cost of the sequence satisfies

\[
\hat{c} \;\ge\; \sum_{j=1}^k \mathrm{range}(B_j)
\qquad\text{and}\qquad
2\hat{c} \;\ge\;\sum_{j=1}^{k-1} \mathrm{range}(B_j\cup B_{j+1})
\]
\end{proposition}

Intuitively, this states that the minimum cost of each block is realized when it is sorted, that cost being its range. So, the cost of the algorithm with predictions must have a cost of $\hat{c}$ that is at least the cost of the sum over all blocks of the range within each block. 

\begin{claim}
    \label{claim:consistent}
    Our proposed algorithm is $O(\frac{n}{r})$-consistent.
\end{claim}
\begin{proof}
    \label{proof:consistent}
    We incur costs in two distinct ways: costs within blocks and costs between the edge elements of consecutive blocks. 
    
    First, consider the cost within each block. Since the optimal cost for each block is $\mathrm{range}(B_j)$ and each block has $m$ elements, the cost of the block is at most $C\cdot\sqrt{m}\cdot\mathrm{range}(B_j)$ where $C$ is the constant hidden in the asymptotics from \textsc{Alg}$_0$. Then, the cumulative cost over all blocks is $\sum_{j=1}^k(C\cdot\sqrt{m}\cdot\mathrm{range}(B_j))=C\cdot\sqrt{m}\cdot\sum_{j=1}^k\mathrm{range}(B_j)$. By the first inequality in Proposition \ref{proposition:block-range-bound}, we see that this value is upper bounded by $C\cdot\sqrt{m}\cdot\hat{c}=C\cdot\frac{n}{r}\cdot\hat{c}$. 

    Next, consider the cost between consecutive blocks. For a single pair of blocks, this value is the difference between the right-most value of the left block and the left-most value of the right block. At worst, this value is the difference between the maximum and minimum value across both sets, or $\mathrm{range}(B_j\cup B_{j+1})$. Considering all $k-1$ consecutive pairs of blocks, we get a cost of at most $\sum_{j=1}^{k-1}\mathrm{range}(B_j\cup B_{j+1})$. By the second inequality in Proposition \ref{proposition:block-range-bound}, we know this value is in turn at most $2\hat{c}$. 

    Combining the two distinct types of costs, we see that this algorithm has a cost of at most
    \[
    C\cdot\frac{n}{r}\cdot\hat{c} + 2\hat{c}=O(\frac{n}{r}\cdot\hat{c})
    \]
    Thus, the algorithm is $O(\frac{n}{r})$-consistent. 
\end{proof}

\begin{claim}
    \label{claim:robust}
    Our proposed algorithm is $O(r)$-robust. 
\end{claim}
\begin{proof}
    \label{proof:robust}
    Again, consider the costs within each block first. By our use of \textsc{Alg}$_0$, we know the cost of any block is at most $C\cdot\sqrt{m}\cdot c_{opt}'$, where $c_{opt}'$ is the optimal cost for the given block. It is easy to see that $c_{opt}'\leq c_{opt}$. The cost across all $k$ blocks is then at most $k\cdot(C\cdot\sqrt{m}\cdot c_{opt})=C\cdot\frac{r^2}{n}\cdot\frac{n}{r}\cdot c_{opt}=C\cdot r\cdot c_{opt}$. 

    Considering the cost between pairs of consecutive blocks, we know that the maximum cost between two elements is $c_{opt}$. So, the cost for all $k-1$ consecutive pairs in adjacent blocks is at most $k\cdot c_{opt}=\frac{r^2}{n}\cdot c_{opt}$. Since $r\leq n$, we can see $\frac{r}{n}\leq 1$, and in turn $\frac{r^2}{n}\leq r$. Then the cost between consecutive blocks is at most $r\cdot c_{opt}$. 

    Once again, combining the two distinct types of costs, we see that this algorithm has a cost of at most
    \[
    C\cdot r\cdot c_{opt} + r\cdot c_{opt}=(C+1)\cdot r\cdot c_{opt}=O(r\cdot c_{opt})
    \]
    Thus, the algorithm is $O(r)$-robust. 
\end{proof}

\section{Online Metric TSP With Predictions}
\label{sec:tsp}

The online sorting problem is a special case of the online traveling salesman problem (online TSP) on the line metric \cite{abrahamsen2024online, bercea2026online}. In the general setting, the input sequence $X = x_1, \dots, x_n$ is drawn from an arbitrary metric space $(\mathcal{U}, d)$, and the cost of an array $A$ is defined as $c(A) = \sum_{i=1}^{n-1} d(A[i], A[i+1])$. 

Our algorithm from Section~\ref{sec:algorithm} extends directly to general metric spaces without knowing $c_{opt}$. We maintain an identical block partition $B_1, \dots, B_k$ of width $m = (\frac{n}{r})^2$ and routing rule based on $\hat{i}_t$, and simply replace \textsc{Alg}$_0$ inside each block with the deterministic $O(\sqrt{n})$-competitive algorithm for online metric TSP by Bertram~\cite{bertram2025online}, which we denote as \textsc{Alg}$_1$. Like \textsc{Alg}$_0$, \textsc{Alg}$_1$ does not require prior knowledge of the optimal cost and does not deal with predictions.

\subsection{Cost Analysis}
To adapt the analysis of Section~\ref{sec:algorithm} to a general metric $d$, let $\hat{A}$ denote the array formed by following the predictions directly, and for any block $B_j$, let $c(\hat{A}[B_j])$ and $c_{opt}(B_j)$ denote the predicted cost and the optimal TSP cost, respectively, restricted to the elements assigned to $B_j$. Finally, let $\mathrm{diam}(S) = \max_{x, y \in S} d(x, y)$ denote the diameter of a subset of elements $S$. We use the metric analog of Proposition~\ref{proposition:block-range-bound} from \cite{bercea2026online}.

\begin{proposition}[Adapted from Bercea et al.~\cite{bercea2026online}, Claim 17]
\label{proposition:tsp-bound}
Let $A$ be an array partitioned into $k$ blocks $B_1, \dots, B_k$. Then the predicted cost $\hat{c} = c(\hat{A})$ satisfies
\[
\hat{c} \;\ge\; \sum_{j=1}^k c(\hat{A}[B_j]) \;\ge\; \sum_{j=1}^k \mathrm{diam}(B_j),
\qquad\text{and}\qquad
2\hat{c} \;\ge\; \sum_{j=1}^{k-1} \mathrm{diam}(B_j \cup B_{j+1}).
\]
\end{proposition}

\begin{claim}
\label{claim:tsp}
For any $\sqrt{n} \le r \le n$, the proposed algorithm for online metric TSP with predictions is $O(\frac{n}{r})$-consistent and $O(r)$-robust without knowledge of $c_{opt}$.
\end{claim}
\begin{proof}
Because the block sizes $m = (\frac{n}{r})^2$ and block count $k = \frac{r^2}{n}$ are identical to Section~\ref{sec:algorithm}, it suffices to show that the intra-block and inter-block metric costs obey the same upper bounds used in Claims~\ref{claim:consistent} and \ref{claim:robust}.

For consistency, running \textsc{Alg}$_1$ on block $B_j$ guarantees a cost within the block of at most $C' \cdot \sqrt{m} \cdot c_{opt}(B_j)$. Because $\hat{A}[B_j]$ is a valid ordering of the exact multiset of elements assigned to $B_j$, its cost is at least the optimal cost on those elements, so $c_{opt}(B_j) \le c(\hat{A}[B_j])$. Summing over all $k$ blocks and applying the first inequality of Proposition~\ref{proposition:tsp-bound} gives
\[
\sum_{j=1}^k C' \cdot \sqrt{m} \cdot c_{opt}(B_j) \;\le\; C' \cdot \sqrt{m} \sum_{j=1}^k c(\hat{A}[B_j]) \;\le\; C' \cdot \frac{n}{r} \cdot \hat{c}.
\]
Furthermore, the distance between the adjacent boundary elements of $B_j$ and $B_{j+1}$ is at most $\mathrm{diam}(B_j \cup B_{j+1})$, which sums to at most $2\hat{c}$ across all $k-1$ boundaries by the second inequality of Proposition~\ref{proposition:tsp-bound}. The remainder of the consistency bound follows identically to Claim~\ref{claim:consistent}.

For robustness, observe that by the triangle inequality, shortcutting the globally optimal tour past all elements not in $B_j$ cannot increase its cost, which implies $c_{opt}(B_j) \le c_{opt}$. Likewise, the distance between any pair of boundary elements in adjacent blocks is at most $\mathrm{diam}(X) \le c_{opt}$. Substituting these two bounds into the proof of Claim~\ref{claim:robust} immediately yields a total cost of at most $(C' + 1) \cdot r \cdot c_{opt} = O(r \cdot c_{opt})$.
\end{proof}
\section{Closing Remarks}
\label{sec:conclusion}
In this work we improve on prior algorithms that aim to leverage learned predictions in the setting of online sorting and its generalization to online metric TSP. Specifically, we provide an algorithm for each, inspired heavily by previous work, that relieves the constraint of knowing the optimal cost before all elements of the online input have been revealed. 

We leave open questions that reiterate those of prior work. In particular, we are interested in predictions that need not produce a permutation of the input. A seemingly significant question is whether we can achieve the same bounds when there are collisions in the predictions. Can we handle such collisions in an efficient way?

\bibliographystyle{plain}
\bibliography{references}

\end{document}